\documentclass[letterpaper,twocolumn]{IEEEAerospaceCLS}  
\usepackage{style}

\usepackage[]{graphicx}    
\newcommand{\ignore}[1]{}  

\usepackage{graphicx}
\usepackage{amsmath}
\usepackage[version=4]{mhchem}
\usepackage{siunitx}
\usepackage{longtable,tabularx}
\usepackage{indentfirst}
\usepackage{float}

\begin{document}
\title{Information-Optimal Multi-Spacecraft Positioning for Interstellar Object Exploration}

\author{%
\large{Arna Bhardwaj$^*$~~\quad Shishir Bhatta$^*$~~\quad Hiroyasu Tsukamoto$^\dag$}\vspace{0.7em}\\
\normalsize{Department of Aerospace Engineering, University of Illinois at Urbana-Champaign, Urbana, IL}
\thanks{$^*$ Undergraduate Researcher, Department of Aerospace Engineering, University of Illinois at Urbana-Champaign, Urbana, IL, {\tt\footnotesize\{\href{mailto:arna2@illinois.edu}{arna2}, \href{mailto:sb74@illinois.edu}{sb74}\}@illinois.edu}. Assistant Professor, Department of Aerospace Engineering, University of Illinois at Urbana-Champaign, Urbana, IL, {\tt\footnotesize\href{mailto:hiroyasu@illinois.edu}{hiroyasu@illinois.edu}}. Part of the research was carried out at the Jet Propulsion Laboratory, California Institute of Technology, under a contract with the National Aeronautics and Space Administration.}}

\maketitle

\thispagestyle{titlepagefancy}
\pagestyle{myfancy}

\begin{abstract}
Interstellar objects (ISOs), astronomical objects not gravitationally bound to the sun, could present valuable opportunities to advance our understanding of the universe's formation and composition. In response to the unpredictable nature of their discoveries that inherently come with large and rapidly changing uncertainty in their state, this paper proposes a novel multi-spacecraft framework for locally maximizing information to be gained through ISO encounters with formal probabilistic guarantees. Given some approximated control and estimation policies for fully autonomous spacecraft operations, we first construct an ellipsoid around its terminal position, where the ISO would be located with a finite probability. The large state uncertainty of the ISO is formally handled here through the hierarchical property in stochastically contracting nonlinear systems. We then propose a method to find the terminal positions of the multiple spacecraft optimally distributed around the ellipsoid, which locally maximizes the information we can get from all the points of interest (POIs). This utilizes a probabilistic information cost function that accounts for spacecraft positions, camera specifications, and ISO position uncertainty, where the information is defined as visual data collected by cameras. Numerical simulations demonstrate the efficacy of this approach using synthetic ISO candidates generated from quasi-realistic empirical populations. Our method allows each spacecraft to optimally select its terminal state and determine the ideal number of POIs to investigate, potentially enhancing the ability to study these rare and fleeting interstellar visitors while minimizing resource utilization.


%

\end{abstract}

\tableofcontents

\section{Introduction}
Interstellar objects (ISOs), astronomical objects traveling through space without any attachment to any system, present a unique opportunity to study various aspects of the universe, such as the formation and composition of other star systems, the origin of the universe, and possibly other forces at play in the expanses of space \cite{bannister_natural_2019}. For example, one of the two ISOs observed to date~\cite{oumuamua,guzick-2020}, named 1I/`Oumuamua, visited us from the rough direction of the constellation Lyra exhibiting some unusual physical characteristics (a highly elongated shape, a lack of typical cometary volatiles, and a deviation from Keplerian trajectories~\cite{Seligman_2021,dybczynski_investigating_2018}), which might provide some clues into how our solar system and neighboring exoplanetary star systems are formed. Exploring these interstellar visitors in situ with dedicated spacecraft would offer an alternative means to acquire firsthand knowledge about interstellar space, which goes beyond what is obtained through remote observations via telescopes.

Due to the hyperbolic nature of their orbit, ISOs only pass through the solar system once in their lifetime with inherently large and rapidly changing uncertainty in their state, often at high inclination and relative velocities~\cite{my_iso_mission}. This poses significant challenges in designing a guidance, navigation, and control (GNC) strategy for the ISO encounter, requiring fast response autonomous operations even with the limited computational capacity of spacecraft. These would involve some levels of approximations in the GNC policies for the sake of their onboard execution, which then introduces additional difficulties in formally quantifying the probability of a successful encounter. Furthermore, theoretical connections between this probability and our mission objective of maximizing the scientific return through the encounter (which we characterize by the amount of visual data collected by spacecraft cameras) still remain ambiguous.
\subsection{Contributions}
This paper proposes a novel multi-spacecraft framework for locally maximizing the amount of information obtained through ISO encounters, with some probabilistic guarantees. Building on our previous work~\cite{myiso}, we first derive a formal upper bound on the failure probability of the ISO encounter for hierarchical It\^{o} stochastic nonlinear system, which models our spacecraft relative dynamics in the local-vertical local-horizontal (LVLH) frame~\cite[pp. 710-712]{lvlhbook} centered on the ISO with the large uncertainty in its state measurement. In order to achieve fast response autonomous operations even with the limited onboard computational capacity of the spacecraft, potential approximation errors in their control and estimation policies are explicitly considered when deriving the bound. The failure probability then enables the construction of an ellipsoid around the terminal position of the chief spacecraft, in which the ISO would be located at the terminal time with a finite probability. We further propose an approach to finding the terminal positions of the deputy spacecraft with respect to the chief spacecraft, optimally distributed around the ellipsoid to maximize the information we get from all the points of interest (POIs) in it. In particular, we utilize an information cost function characterized by the distance between the center of the ellipsoid and each spacecraft in the swarm and the visual coverage of the POIs, accounting for spacecraft positions and attitudes, camera specifications, and ISO position uncertainty, where the information here is defined as imaging data collected by spacecraft onboard cameras. Since this optimization only requires an upper bound for the ISO state uncertainty history over time, which we could obtain as in our previous work~\cite{my_iso_mission} with~\cite{declan_nav,autonav}, the optimal relative positions of the deputy spacecraft can be pre-computed offline, freeing onboard computational resources for some other highly autonomous tasks.

Numerical simulations are performed to demonstrate the efficacy of our method using a quasi-realistic empirical population of ISOs~\cite{Engelhardt_2017,declan_nav} with the empirical history of the navigation uncertainties~\cite{declan_nav,autonav,my_iso_mission,myiso}, where the optimization problems are solved using Nelder-Mead. Our first simulation examines the probability that a single spacecraft would be able to view the ISO, assuming an uncertainty ellipsoid with equal $x$, $y$, and $z$ radii. As is expected from the aforementioned probability bound, the probability that the spacecraft views the ISO decreases as the ISO state uncertainty increases. Our second simulation empirically determines the optimal number of spacecraft to view as many POIs on the sphere as possible, where the size of the uncertainty sphere remains constant. It is demonstrated that a multi-spacecraft system both observes more POIs and results in a lower information cost than a singular spacecraft, however, an excess of spacecraft increases the overlap between the visual coverage of the system. 
\subsection{Related Work}
Previous studies have investigated the potential of ISO flyby missions for a single spacecraft, proposing viable methods for observing ISOs and other high-speed celestial objects~\cite{my_iso_mission,myiso}. For example, it is rigorously shown that a control policy designed using the spectrally normalized deep neural network (SN-DNN) guarantees local, finite-time exponential boundedness of the expected spacecraft delivery error, assuming that the bounds for the estimation errors are known. These approaches would allow at least to encounter ISOs even under their large state uncertainty. However, the entirety of the ISO cannot be observed through a flyby as matching the spacecraft velocity with the ISOs' is unrealistic with current propulsion technologies~\cite{declan_nav}, making such missions inefficient for gathering valuable scientific data. Several numerical and theoretical methods have been proposed to circumvent these difficulties through the lens of information-based optimization. Previous works in this field with a swarm of spacecraft can be broken into two categories: gaining information from a body with a determined orbit, such as satellites or other spacecraft, and gaining information from an object with an undetermined orbit, such as celestial bodies.

In~\cite{Attitude-Control-SC-Bodies}, a passive observation technique is proposed for celestial bodies with known trajectories, whereas for the case of ISOs, an active approach should be taken to sufficiently capture information to address the large state uncertainty. Further, the algorithm utilized in the study only maximizes the amount of the celestial body covered by the field of view (FOV) of the spacecraft in the swarm. The overlap between spacecraft in the swarm must be minimized to make a more robust algorithm. This would prevent resources from being wasted by spacecraft redundantly examining the same area of the ISO. The collision avoidance techniques introduced in~\cite{Distributed-fast-MP} utilize a minimum scalar distance that a spacecraft must be from each other. However, this constraint ignores any angular or directional constraints. With the highly uncertain trajectory of the ISO, it is crucial to minimize the amount of overlap in information gained by any two spacecraft.

This consideration is crucial when optimizing information from bodies with controlled trajectories, such as during on-orbit inspection or servicing of satellites \cite{Ref:Info-based-GandC}, \cite{Ref:Coord-Motion-Planning}. These studies utilize an information cost to create a fuel and energy-efficient orbit around a target satellite. However, the information cost presented in these studies does not account for the spacecraft's orientation or field of view. Instead, they determine a baseline distance for the spacecraft to start at and then algorithmically determine an orbit for the spacecraft to take. In the case of this study, orbiting the ISO would be infeasible again due to the ISO's high relative velocity and the uncertainty of its dynamics. Some other studies focus on optimizing information from celestial bodies with uncontrolled trajectories; however, only comets and asteroids have been examined so far \cite{Ref:Attitude-control,Ref:Complex-Space-Structures}. In these studies, either the spacecraft that have to gain information are stationary, or the object they are observing is stationary and the spacecraft performs a flyby inspection of the celestial body. With an ISO, the spacecraft has to assume a formation in order to maximize the information gained, since the ISO's state is uncertain and it never remains stationary. This paper proposes one solution to the issues listed above through the proactive deployment of multiple spacecraft, extending the ideas of information gain-based optimization. Hierarchical stochastic contraction also enables explicitly considering the state uncertainty of the ISO, having the spacecraft cooperate to gather visual data on all perspectives of the ISO locally optimally in one flyby mission.


%
\section{Problem Formulation}
\label{sec_problem}
This paper considers the following translational dynamics of a chief spacecraft (which could be virtual) relative to an Interstellar Object (ISO):
\begin{subequations}
    \label{original_dyn}
\begin{align}
    \begin{bmatrix}
        d\toe \\
        dx
    \end{bmatrix}
    &= 
    \begin{bmatrix}
        \psi(\toe,t) \\
        f(\toe,x,t)+g(\toe,x,t)u_a(\htoe,\hx,t)
    \end{bmatrix}dt \label{eq_dyn}\textcolor{white}{\eqref{eq_dyn}} \\
    \begin{bmatrix}
        d\htoe \\
        d\hx
    \end{bmatrix}
    &= 
    \begin{bmatrix}
        \psi(\htoe,t) \\
        f(\htoe,\hx,t)+g(\htoe,\hx,t)u_a(\htoe,\hx,t)
    \end{bmatrix}dt \label{eq_estimation}\textcolor{white}{\eqref{eq_estimation}} \\
    &+\ell_a(\htoe,\hx,t) (y-h(\htoe,\hx,t))dt \\
    ydt &= h(\toe,x,t)dt+G(\toe,x,t)d\mathcal{W}\label{eq_measurement}\textcolor{white}{\eqref{eq_measurement}}
\end{align}    
\end{subequations}
with the following dynamics of $N$ deputy spacecraft ($i = 1,\cdots,N$) relative to the chief spacecraft: 
\begin{align}
    \label{original_dyn_i}
    dx_i &= (f_i(\toe_c,u_c,x_i,t)+g_i(\toe_c,x_i,t)u_i(\toe_c,u_c,x_i,t))dt
\end{align}
where the relative state of the chief spacecraft is given in the local-vertical local-horizontal (LVLH) frame centered on the ISO~\cite[pp. 710-712]{lvlhbook}, the relative state of the deputy spacecraft is given in the LVLH frame centered on the chief spacecraft, $\psi$, $f$, $g$, $h$, $f_i$, and $g_i$ are known smooth functions (see, \eg{},~\cite{doi:10.2514/1.G000218,doi:10.2514/1.55705,doi:10.2514/1.37261}), ${\mathcal{W}}:{\mathbb{R}}_{\geq 0} \mapsto {\mathbb{R}}^{w}$ is a $w$-dimensional Wiener process for representing the measurement noise with a known matrix-valued function $G$, and the definitions of the other symbols are given in Table~\ref{tab_dynamics}.

Note that (a) the dynamics are expressed using the It\^{o} stochastic differential equation accounting for the stochastic uncertainty of the state measurement~\cite[p. 100]{arnold_SDE} (see also~\cite[p. xii]{arnold_SDE} for the notations used); (b) the process and the measurement noises other than $G(\toe,x,t)d\mathcal{W}$ are omitted in the chief and deputy spacecraft dynamics of~\eqref{original_dyn}~and~\eqref{original_dyn_i} for simplicity of our discussion, assuming that they are negligible compared to the dominant noise in the measurement $y$ due to the large state uncertainty of the ISO; and (c) the orbital elements $\toe_c$ and the control input $u_c$ of the chief spacecraft are assumed to be given to all the deputy spacecraft via communication and relative navigation.
\begin{table}[htbp]
\caption{Summary of the symbols in~\eqref{original_dyn}~and~\eqref{original_dyn_i}.}
\label{tab_dynamics}
\scriptsize
\begin{center}
\renewcommand{\arraystretch}{1.3}
\begin{tabular}{ l l }
\hline
\hline
{Symbol} & {Description} \\ 
\hline
S/C & Shorthand for spacecraft  \\
$\toe$ & ISO orbital elements, $\toe: \mathbb{R}_{\geq 0} \to \mathbb{R}^n$ \\ 
$x$ & State of chief S/C relative to $\toe$, $x: \mathbb{R}_{\geq 0} \to \mathbb{R}^n$ \\ 
$\htoe$ & Estimated ISO state, $\htoe: \mathbb{R}_{\geq 0} \to \mathbb{R}^n$ \\ 
$\hx$ & Estimated state of $x$, $\hx: \mathbb{R}_{\geq 0} \to \mathbb{R}^n$ \\
$u_a$ & Controller of chief S/C, $u_a:\mathbb{R}^n\times\mathbb{R}^n\times{\mathbb{R}}_{\geq 0}\mapsto{\mathbb{R}}^m$ \\
$\ell_a$ & Estimation gain, $\ell_a:\mathbb{R}^n\times\mathbb{R}^n\times{\mathbb{R}}_{\geq 0}\mapsto \to \mathbb{R}^{2n\times k}$ \\
$y$ & State measurement of chief S/C, $y: \mathbb{R}_{\geq 0} \to \mathbb{R}^k$\\
$i$ & Index of deputy S/C, $i = 1, ..., N$ \\
$x_i$ & State of $i$th deputy S/C relative to $x$, $x_i: \mathbb{R}_{\geq 0} \to \mathbb{R}^n$ \\ 
$\toe_c$ & Orbital elements of chief S/C, $\toe_c: \mathbb{R}_{\geq 0} \to \mathbb{R}^n$  \\ 
$u_c$ & Shorthand for $u_a(\htoe,\hx,t)$ communicated to deputy S/C  \\ 
$u_i$ & Controller of deputy S/C, $u_i:\mathbb{R}^n\times\mathbb{R}^m\times\mathbb{R}^n\times{\mathbb{R}}_{\geq 0}\mapsto{\mathbb{R}}^m$ \\
\hline
\hline
\end{tabular}
\end{center}
\end{table}
\subsection{Technical Challenges in ISO Exploration}
Unlike conventional exploration targets with sufficient information on their orbital properties in advance, ISOs always come with the following inherent GNC challenges due to their poorly constrained orbits with generally high inclinations and relative velocities:
\begin{enumerate}
    \item The ISO state and its onboard estimate in~\eqref{original_dyn} change dramatically in time due to the dominant stochastic term $G(\toe,x,t)d\mathcal{W}$, where using pre-designed GNC policies becomes unrealistic even at the terminal phase of the mission. This leads to significant onboard resources for uncertainty quantification and utilization with the radically changing ISO state information.\label{itemissue1}
    \item Due to the limited onboard computational capacity of the spacecraft, executing complex algorithms can be challenging. This is especially true for small satellites like femtosatellites, typical in affordable multi-spacecraft missions, where running any nonlinear and optimal GNC algorithms proactively accounting for the large ISO uncertainty could be nearly impossible.\label{itemissue2}
\end{enumerate}

These issues could be partially addressed for a single spacecraft, \eg{}, by viewing $u_a$ of~\eqref{eq_dyn} as an approximated autonomous control policy of a target controller $u_d$, which could be computationally expensive and gives the following target trajectory with the true states:
\begin{align}
    \label{target_dyn}
    dx_d = (f(\toe,x_d,t)+g(\toe,x_d,t)u_d(\toe,x_d,t))dt
\end{align}
where it is rigorously shown that $u_a$ designed using the techniques of~\cite{myiso} guarantees the local, finite-time exponential boundedness of the approximated trajectory $x$ of~\eqref{eq_dyn} and the target trajectory $x_d$ of~\eqref{target_dyn}, assuming that the bounds for the estimation errors $\htoe$ and $\hx$ in~\eqref{eq_estimation} are known.

In this paper, we extend the ideas of this result further by viewing both $u_a$ of~\eqref{eq_dyn} and $\ell_a$ of~\eqref{eq_estimation} as approximations of a target control policy $u_d$ and target estimation gain $\ell_d$ that guarantee contraction of the system~\eqref{original_dyn}~\cite{Ref:contraction1}. With some reasonable assumptions, we first explicitly compute the failure probability of the chief spacecraft's ISO encounter under these approximated policies, extensively using the hierarchical contraction in stochastic systems~\cite{Ref:contraction1,Pham2009}. The failure probability is used to construct an ellipsoid around the terminal position of the chief spacecraft, in which the target ISO would be located at the terminal time with a finite probability. We then propose a method to find the terminal positions of the deputy spacecraft optimally distributed around the ellipsoid, which maximizes the information we can get out of all the points of interest (POIs) in the ellipsoid. In particular, we utilize an information cost function characterized by the distance between the POIs and each spacecraft in the swarm, accounting for spacecraft positions and attitudes, camera specifications, and ISO position uncertainty so we can get sufficient visual data on the entirety of the ISO. 

The final part of the deputy spacecraft's optimization can be performed as long as we have an upper bound for the uncertainty history over time, $\|G(x(t),t)\|$ for $G$ of~\eqref{eq_measurement} (which we empirically obtained in our previous work~\cite{my_iso_mission} using the AutoNav system~\cite{autonav}) and thus can be pre-computed.
\section{Probability of ISO Encounter}
Suppose that our controller $u_a$ of~\eqref{eq_dyn} and estimation gain $\ell_a$ of~\eqref{eq_estimation} approximate some target control policy $u_d$ and target estimation gain $\ell_d$ of~\eqref{eq_estimation} to avoid the potential computational burden associated with the implementation of $u_d$ and $\ell_d$. Using, \eg{}, the spectrally-normalized deep neural network (SN-DNN), we can ensure that there exist $\epsilon_{u},\epsilon_{e} \in \mathbb{R}_{>0} $ that satisfy the following at least locally~\cite{miyato2018spectral,myiso}:
\begin{align}
    \label{eq_approx_errors}
    \|u_a(\toe,x,t) - u_d(\toe,x,t)\| \leq \epsilon_{u},~\|\ell_a(\toe,x,t) - \ell_d(\toe,x,t)\| \leq \epsilon_{u}
\end{align}
for all $(\toe,x,t) \in \mathcal{S}$, where $\mathcal{S}$ is some compact set containing $(\toe,x,t)$. Also, $u_a(\htoe,\hx,t)$ and $\ell_a(\htoe,\hx,t)$ of~\eqref{original_dyn} can be decomposed into the estimation error and approximation error parts as follows:
\begin{align}
    u_a(\htoe,\hx,t) &= \underbrace{u_a(\htoe,\hx,t)-u_a(\toe,x,t)}_{\text{estimation error part}} \\
    &+\underbrace{u_a(\toe,x,t)-u_d(\toe,x,t)}_{\text{approximation part}}+u_d(\toe,x,t) \\
    \ell_a(\htoe,\hx,t) &= \underbrace{\ell_a(\htoe,\hx,t)-\ell_d(\htoe,\hx,t)}_{\text{approximation error part}}+\ell_d(\htoe,\hx,t)
\end{align}
where $\toe$, $x$, $\htoe$, and $\hx$ are the states of~\eqref{original_dyn}. Based on these observations, let us define smooth paths $q_c(\mu,t)$ and $q_e(\mu,t) = [q_{e_1}(\mu,t)^{\top},q_{e_2}(\mu,t)^{\top}]^{\top}$ parameterized by $\mu \in [0,1]$ to have
\begin{itemize}
    \item $q_c(0,\mu)=x_d$, $q_{e_1}(0,\mu)=\toe$, and $q_{e_2}(0,\mu)=x$
    \item $q_c(1,\mu)=x$, $q_{e_1}(1,\mu)=\htoe$, and $q_{e_2}(1,\mu)=\hx$
\end{itemize}
where $x_d$ is the target trajectory of~\eqref{target_dyn}, and consider the following virtual system of $q_c(\mu,t)$ and $q_e(\mu,t)$~\cite{Ref:contraction3,tutorial}:
\begin{subequations}
\label{eq_virtual_system}
\begin{align}
    dq_c &= \underbrace{f(\toe,q_c,t)+g(\toe,q_c,t)u_d(\toe,q_c,t)}_{\text{contracting part}}dt \label{eq_virtual_control}\textcolor{white}{\eqref{eq_virtual_control}}\\
    &\hspace{-1.0em}+g(\toe,x,t)\underbrace{(u_a(q_{e_1},q_{e_2},t)-u_a(\toe,x,t))}_{\text{estimation error part}}dt \\
    &\hspace{-1.0em}+\mu g(\toe,x,t)\underbrace{(u_a(\toe,x,t)-u_d(\toe,x,t))}_{\text{approximation error part}}dt \\
    \begin{bmatrix}
        dq_{e_1} \\
        dq_{e_2}
    \end{bmatrix}
    &=
    \underbrace{
    \begin{bmatrix}
        \psi(q_{e_1},t) \\
        f(q_{e_1},q_{e_2},t)+g(q_{e_1},q_{e_2},t)u_a(\htoe,\hx,t)
    \end{bmatrix}dt}_{\text{contracting part}} \label{eq_virtual_estimation}\textcolor{white}{\eqref{eq_virtual_estimation}} \\
    &\hspace{-1.0em}\underbrace{-\ell_d(\htoe,\hx,t) (h(q_{e_1},q_{e_2},t)-h(\toe,x,t))dt}_{\text{contracting part}} \nonumber \\
    &\hspace{-1.0em}+\underbrace{(\ell_d(\htoe,\hx,t) - \ell_a(\htoe,\hx,t))}_{\text{approximation error part}}\underbrace{(h(q_{e_1},q_{e_2},t)-h(\toe,x,t))}_{\text{estimation error part}}dt \nonumber \\
    &\hspace{-1.0em}+\underbrace{\mu \ell_a(\htoe,\hx,t)G(\toe,x,t)d\mathcal{W}}_{\text{measurement noise}}
\end{align}
\end{subequations}
where the definition of the contracting dynamics is to be elaborated in Theorem~\ref{thm_hierarchical_contraction}. It can be verified that setting $\mu=0$ in~\eqref{eq_virtual_control} and \eqref{eq_virtual_estimation} results in~\eqref{target_dyn} and~\eqref{eq_dyn}, respectively, and setting $\mu=1$ results in the $x$ dynamics of~\eqref{eq_dyn} and~\eqref{eq_estimation}, respectively. Consequently, $q_c=x_d$ \& $q_e = [\toe^{\top},x^{\top}]^{\top}$ and $q_c = x$ \& $q_e = [\htoe^{\top},\hx^{\top}]^{\top}$ are indeed particular solutions of~\eqref{eq_virtual_system}.

Let us denote the contracting parts of~\eqref{eq_virtual_control} and~\eqref{eq_virtual_estimation} as $f_c(\toe,q_c,t)$ and $f_e(\xi,q_e,t)$, respectively, where $\xi = (\toe,x,\htoe,\hx)$ is introduced here for keeping the notation compact. If all the error and noise terms in~\eqref{eq_virtual_system} are zero, the differential dynamics of~\eqref{eq_virtual_control} and~\eqref{eq_virtual_estimation} for~$\partial_{\mu}q_c=\partial q_c/\partial \mu$ and~$\partial_{\mu}q_e=\partial q_e/\partial \mu$ are given as
\begin{align}
\label{eq_differential_dyn_sto}
\partial_{\mu}\dot{q}_c= \frac{\partial f_c}{\partial q_c}\partial_{\mu}q_c,~\partial_{\mu}\dot{q}_e = \frac{\partial f_e}{\partial q_e}\partial_{\mu}q_e
\end{align}
as in the standard discussions in contraction theory~\cite{Ref:contraction1,Pham2009,tutorial}, which gives the following result using the hierarchical property of stochastic contraction.
\begin{theorem}
\label{thm_hierarchical_contraction}
Suppose that we have uniformly bounded positive definite matrices $\underline{m}_c \mathbb{I} \preceq M_c(t) \preceq \overline{m}_c \mathbb{I}$ and $\underline{m}_e \mathbb{I} \preceq M_e(t) \preceq \overline{m}_e \mathbb{I}$ that satisfy the following contraction conditions:
\begin{subequations}
\label{eq_contraction}
\begin{align}
    &\dot{M}_c + M_c\frac{\partial f_c}{\partial x_c} + \frac{\partial f_c}{\partial x_c}^{\top}M_c \leq -2\alpha_c M_c \label{contraction_c}\textcolor{white}{\eqref{contraction_c}} \\
    &\dot{M}_e + M_e\frac{\partial f_e}{\partial x_e} + \frac{\partial f_e}{\partial x_e}^{\top}M_e \leq -2\alpha_e M_e \label{contraction_e}\textcolor{white}{\eqref{contraction_e}} 
\end{align}    
\end{subequations}
where $\alpha_c,\alpha_e,\underline{m}_c,\overline{m}_c,\underline{m}_e,\overline{m}_e\in\mathbb{R}_{> 0}$. Suppose also that the approximation errors are bounded as in~\eqref{eq_approx_errors} and that $\exists\bar{g},\bar{u},\bar{h},\bar{\ell}\in\mathbb{R}_{\geq 0}$ and $\bar{\zeta}:\mathbb{R}_{\geq0}\mapsto\mathbb{R}_{\geq 0}$ \st{} $\|g(\toe,x,t)\| \leq \bar{g}$, $\|\partial u_a/\partial q_e\| \leq \bar{u}$, $\|\partial h/\partial q_e\| \leq \bar{h}$, $\|\ell_a(\htoe,\hx,t)\|_F^2 \leq \bar{\ell}$, $\|G(\toe,x,t)\|_F^2 \leq \bar{\zeta}(t)$, $\forall \xi,t$, where $\|\cdot\|_F$ denotes the Frobenius norm. Then there exist some positive constants $\alpha_s,\gamma_c,\lambda \in \mathbb{R}_{>0}$ that satisfy the following probability bound locally for the states and time in $\mathcal{S}$ of~\eqref{eq_approx_errors}:
\begin{align}
    \label{eq_probability_bound}
    \mathbb{P}[\|x(t)-x_d(t)\| \geq D] \leq \frac{\mathbb{E}[\mathcal{V}(0)]e^{-2\alpha_s t}+c_s+e^{-2\alpha_s t}\bar{\zeta}_I(t)}{D\underline{m}} 
\end{align}
where $D$ is a user-defined failure distance,  $c_s ={(\overline{m}_c\bar{g}\epsilon_c)^2}/{(2\alpha_s\gamma_c)}$, $\bar{\zeta}_I(t) = \lambda\overline{m}_e\bar{\ell}\int_{0}^te^{2\alpha_s \tau} \bar{\zeta}(\tau) d\tau$, $\underline{m} = \underline{m}_c+\lambda\underline{m}_e$, and $\mathcal{V}_{ec}$ is the incremental Lyapunov function to be defined in the following proof.
\end{theorem}
\begin{proof}
Let our differential Lyapunov function $V_{ec}$ be defined as
\begin{align}
    V_{ec} = \partial_{\mu}q_c^{\top}M_c\partial_{\mu}q_c+\lambda \partial_{\mu}q_e^{\top}M_e\partial_{\mu}q_e
\end{align}
where $\lambda \in \mathbb{R}_{>0}$ is a constant to be defined later, and the arguments are omitted for notational simplicity.
Applying the infinitesimal differential generator $\mathcal{L}$ of~\cite[p. 15]{sto_stability_book} along with the contraction conditions~\eqref{eq_contraction}, we have
\begin{align}
    \mathcal{L}V_{ec} &\leq -2\alpha_c\underline{m}_c\|\partial_{\mu}q_c\|^2+2\overline{m}_c\|\partial_{\mu}q_c\|(\bar{g}\bar{u} \|\partial_{\mu}q_e\|+\bar{g}\epsilon_c) \\    &-2\lambda\alpha_e\underline{m}_e\|\partial_{\mu}q_e\|^2+2\lambda\overline{m}_e\epsilon_e\bar{h}\|\partial_{\mu}q_e\|^2+\lambda\overline{m}_e\bar{\ell}\bar{\zeta}(t).
\end{align}
Using the relation $2ab \leq \gamma_ca^2+\gamma_c^{-1} b^2$, which holds for any $a,b\in\mathbb{R}$ and any $\gamma_c \in \mathbb{R}_{>0}$, with $a = \|\partial_{\mu}q_c\|$ and $b = \overline{m}_c\bar{g}\epsilon_c$, this reduces to
\begin{align}
    \mathcal{L}V_{ec} &\leq
    \begin{bsmallmatrix}
        \|\partial_{\mu}q_c\| \\
        \|\partial_{\mu}q_e\|
    \end{bsmallmatrix}^{\top}
    \begin{bsmallmatrix}
        -2\alpha_c\underline{m}_c+\gamma_c & \overline{m}_c\bar{g}\bar{u} \\
        \overline{m}_c\bar{g}\bar{u} & -2\lambda(\alpha_e\underline{m}_e-\overline{m}_e\epsilon_e\bar{h})
    \end{bsmallmatrix}    
    \begin{bsmallmatrix}
        \|\partial_{\mu}q_c\| \\
        \|\partial_{\mu}q_e\|
    \end{bsmallmatrix} \\
    &+\gamma_c^{-1}(\overline{m}_c\bar{g}\epsilon_c)^2+\lambda\overline{m}_e\bar{\ell}\bar{\zeta}(t).
\end{align}
Thus, if we select $\alpha_c$, $\alpha_e$, and $\gamma_c$ to have $\exists\bar{\alpha}_c,\bar{\alpha}_e \in \mathbb{R}_{>0}$ \st{} $-2\alpha_c\underline{m}_c+\gamma_c \leq -2\bar{\alpha}_c$ and $-2(\alpha_e\underline{m}_e-\overline{m}_e\epsilon_e\bar{h}) \leq -2\bar{\alpha}_e$, and then select $\lambda$ to have $\alpha_s \in \mathbb{R}_{>0}$ \st{} $\begin{bsmallmatrix}-2\bar{\alpha}_c & \overline{m}_c\bar{g}\bar{u} \\ \overline{m}_c\bar{g}\bar{u} & -2\lambda\bar{\alpha}_e\end{bsmallmatrix} \preceq -2\alpha_s\begin{bsmallmatrix} \overline{m}_c & 0 \\ 0 & \lambda\overline{m}_e\end{bsmallmatrix}$, we get
\begin{align}
    \mathcal{L}V_{ec} &\leq -2\alpha_sV_{ec}+\gamma_c^{-1}(\overline{m}_c\bar{g}\epsilon_c)^2+\lambda\overline{m}_e\bar{\ell}\bar{\zeta}(t).
\end{align}
Applying Dynkin's formula~\cite[p. 10]{sto_stability_book} yields
\begin{align}
    \mathbb{E}[\mathcal{V}_{ec}(t)] \leq (\mathbb{E}[\mathcal{V}_{ec}(0)]+\bar{\zeta}_I(t))e^{-2\alpha_s t}+c_s
\end{align}
where $\mathcal{V}_{ec} = \int_0^1{V}_{ec}d\mu$, $c_s ={(\overline{m}_c\bar{g}\epsilon_c)^2}/{(2\alpha_s\gamma_c)}$, and $\bar{\zeta}_I(t) = \lambda\overline{m}_e\bar{\ell}\int_{0}^te^{2\alpha_s \tau} \bar{\zeta}(\tau) d\tau$. Then the result~\eqref{eq_probability_bound} follows from Markov's inequality~\cite[pp. 311-312]{probbook} with the lower bound of $\mathcal{V}_{ec}(t)$.
\end{proof}
\begin{remark}
It is shown that the positive definite matrices $M_c$ and $M_e$ in~\eqref{eq_contraction} can be found via convex optimization to get the approximated control and estimation policies $u_a$ and $\ell_a$ as in~\eqref{eq_approx_errors}. Reviewing ways to find $M_c$ and $M_e$ in~\eqref{eq_contraction} and making them state-dependent are beyond the scope of this paper, but those interested in knowing more about them can refer to~(\eg{}, in~\cite{Ref:contraction1,myiso,tutorial,Ref:Stochastic} and the references therein.
\end{remark}
\section{Multi-Spacecraft Positioning}
\label{Info:cost}
\begin{figure}[htbp]
    \centering
    \includegraphics[width=0.4\textwidth]{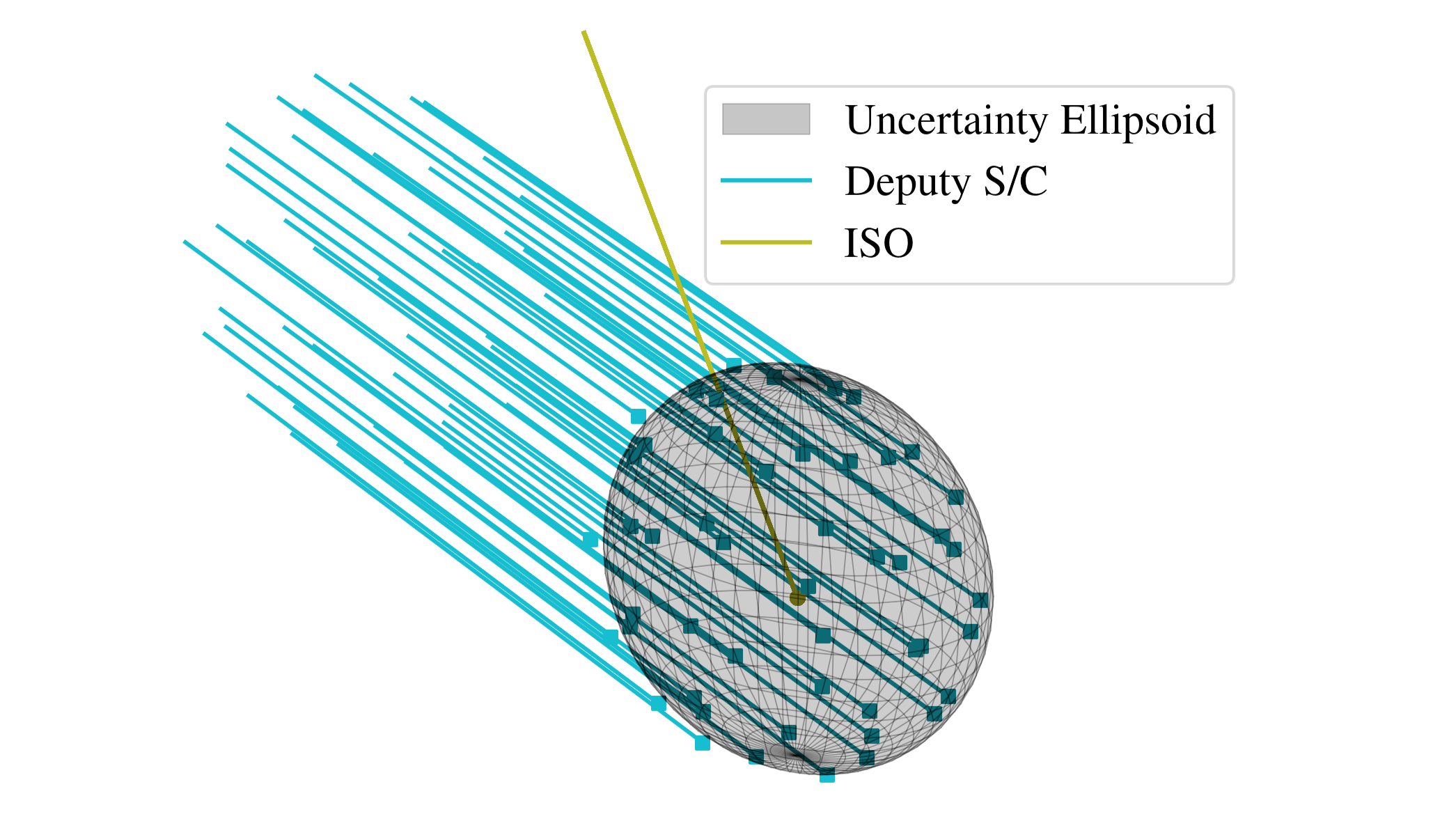}
\vspace{-1em}
    \caption{Conceptual illustration of the deputy spacecraft positioning using Theorem~\ref{thm_hierarchical_contraction} and~\eqref{eq_success}.}
    \label{fig_concept}
\vspace{-2em}
\end{figure}
\begin{figure}[htbp]
    \centering
    \includegraphics[width=0.5\textwidth]{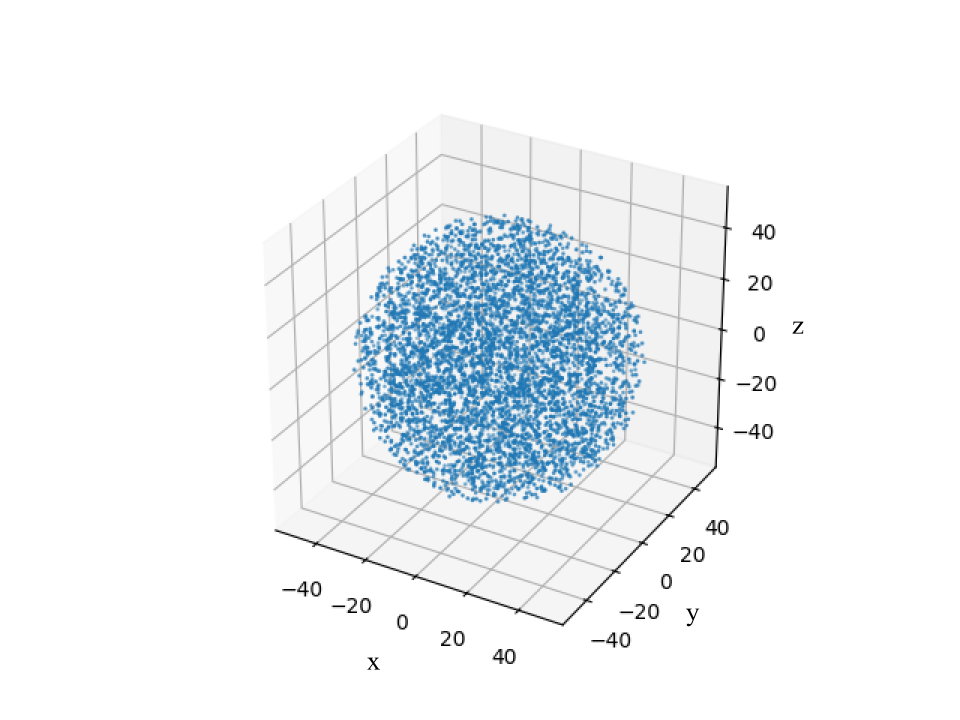}
\vspace{-1em}
    \caption{Example of 5,000 POIs being sampled in an uncertainty ellipsoid.}
    \label{fig:POIs}
\vspace{-3em}
\end{figure}
Now that we have an upper bound of the failure probability of the encounter by Theorem~\ref{thm_hierarchical_contraction}, we can use it to construct an uncertainty sphere around the chief spacecraft, in which the ISO is located at the terminal time $t = T$ with a finite probability. In particular, by setting the desired relative position of the chief spacecraft with respect to the ISO to zero (\ie{}, $p_d(T) = 0$ in~\eqref{eq_probability_bound} for $x_d^{\top} = [p_d^{\top},\dot{p}_d^{\top}]^{\top}$), we get
\begin{align}
    \label{eq_success}
    \mathbb{P}[\|p(T)\|\leq D] \geq \frac{D\underline{m}-\mathbb{E}[\mathcal{V}(0)]e^{-2\alpha_s T}-c_s-e^{-2\alpha_s T}\bar{\zeta}_I(T)}{D\underline{m}}     
\end{align}
where $p(T)$ is the relative position of the chief spacecraft with respect to the ISO at $t = T$. We could also use a weighted norm here to account for the differences in the estimation variance of each state, which results in an ellipsoid.

From the chief spacecraft's point of view, the relation~\eqref{eq_success} implies that as long as all the deputy spacecraft follow the chief spacecraft and are able to observe the entire surface and the interior of the ellipsoid at $t = T$, the ISO will be visible from one of our spacecraft at least with probability $p$, where $p$ is the right-hand side of~\eqref{eq_success}. Also, the chief spacecraft can be virtual as long as we can compute its trajectory leading to the bound~\eqref{eq_probability_bound}. Since the deputy spacecraft uses the relative dynamics~\eqref{original_dyn_i} with respect to the chief spacecraft, whose control input and orbital elements are assumed to be known via communication and relative navigation, existing control and planning strategies, \eg{},~\cite{doi:10.2514/1.G000218,Ref:phasesync}, can be applied without any modifications. However, doing this naively, as illustrated in Fig.~\ref{fig_concept}, would lead to partially wasting the computational resources of the spacecraft by redundantly examining the same area of the ellipsoid. This section discusses an approach to finding locally optimal terminal positions of the deputy spacecraft achieving the coverage, considering practical spacecraft's specifications.

In the following, we assume that POIs are randomly placed within the ellipsoid that the swarm is tasked with viewing, as depicted in Fig.~\ref{fig:POIs}. Since the ISO is expected to be located somewhere within the ellipsoid, the more POIs the swarm views, the more likely they are to view the ISO.
\subsection{Conal Approximation of Field of View}
The field of view (FOV) of a spacecraft is approximated using a cone as in \cite{cone_of_vision}. The vertex of the cone originates at the spacecraft's position $p_{\mathrm{SC}}$ and is oriented at an angle $\theta$ toward the origin of the uncertainty ellipsoid. The radius of the cone is defined as $\nu$ (see Fig.~\ref{fig:cone_fov}). Also, a plane is centered about the origin of the uncertainty ellipsoid and angled toward the cone to divide the ellipsoid into two hemispheres. The points in the hemisphere located further away from the spacecraft are ignored. 
This ensures that each spacecraft can only view one hemisphere of the ellipsoid since a spacecraft can only see one side of an ISO at a time. As described in Fig.~\ref{fig:plane}, if the ISO passes through the center of the uncertainty ellipsoid, the spacecraft can only view the POIs on the side facing toward it. The hemisphere facing away from it must be viewed by another spacecraft. 
\subsection{Cost Function Formulation}
The objective of the information cost is to maximize the amount of POIs visible in each spacecraft's FOV and minimize the total overlap between the FOV of every spacecraft in the swarm. A few remarks are made first about various points that are to be taken into consideration during the formation of the information cost function:
\vspace{-0.5em}
\begin{enumerate}
    \item As the spacecraft approaches the ISO, the information cost should become more optimal due to the spacecraft being able to see more details of the ISO; however, if the spacecraft gets too close, the cameras could fail to capture the entirety of the ISO, capturing only portions of it, so there must be some desired position that provides a compromise to the trade-off.
    \item The FOV of the camera is the angular expanse of the observable environment visible at any instant without any movement \cite{emery_chapter_2017}. It determines how much the camera can see, so orienting the camera such that the ISO lines up within the FOV of the spacecraft will minimize the information costs.
    \item POIs are also randomly sampled throughout the ISO's uncertainty ellipsoid. To quantify the amount of information collected, one of the primary goals of the multi-spacecraft system will be to view as many POIs as possible. By randomly sampling these points, the swarm will attempt to spread out evenly around the uncertainty ellipsoid around the ISO.
    \item In addition to these factors, the system should be concerned about sustainability and optimizing resource consumption. The optimization must minimize the number of spacecraft deployed without significantly decreasing the number of POIs viewed by the system. The system must also consider fuel and energy consumption, selecting terminal positions that require optimizing resource consumption. 
\end{enumerate}

\begin{figure}[htbp]
    \centering
    \includegraphics[width=0.45\textwidth]{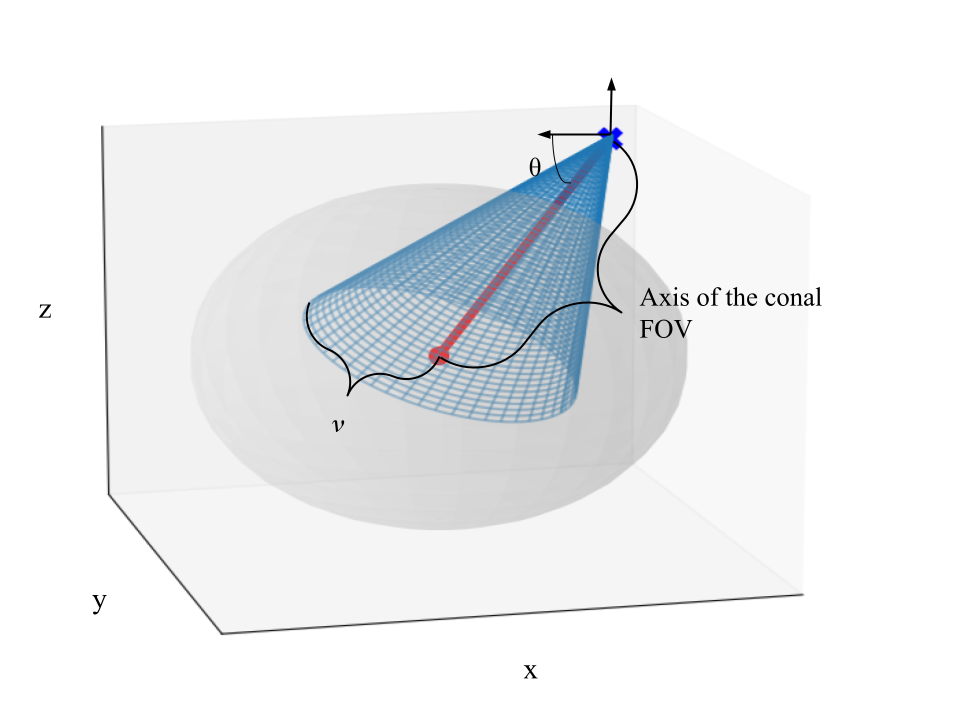}
    \caption{A diagram with a conal FOV with an apex at the spacecraft position is labeled.}
    \label{fig:cone_fov}
\vspace{-2em}
\end{figure}
\begin{figure}[htbp]
    \centering
    \includegraphics[width=0.45\textwidth]{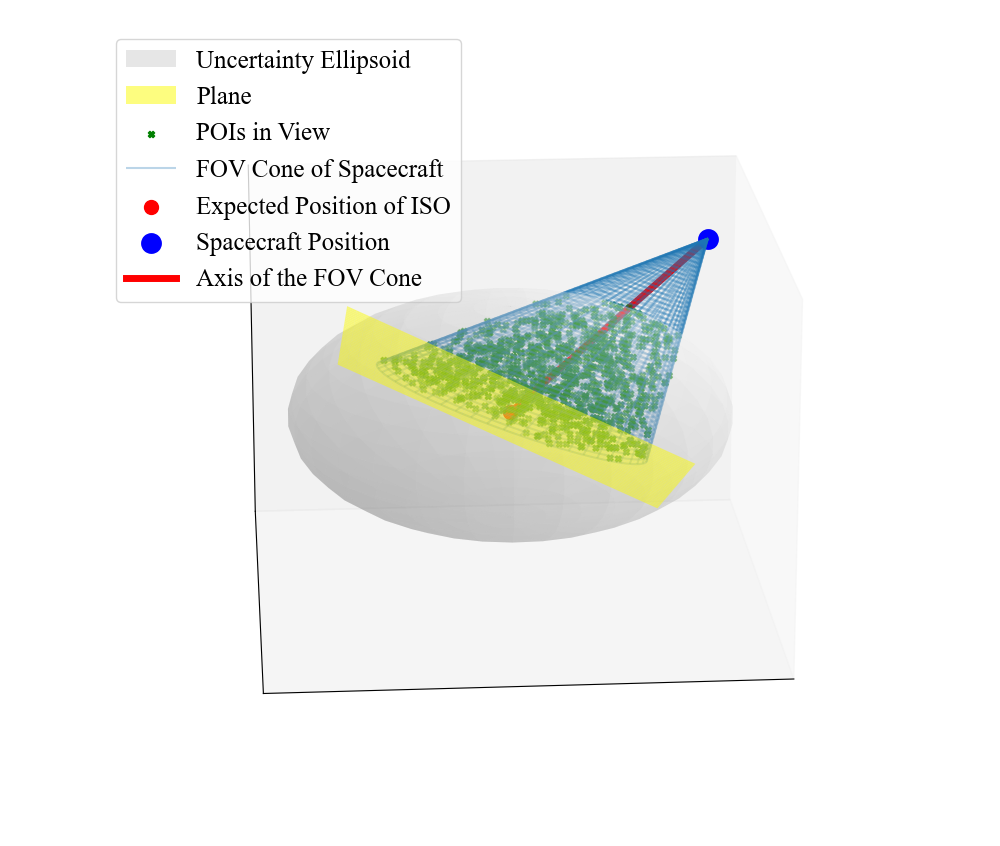}
\vspace{-2em}
    \caption{A diagram of the plane dividing the uncertainty ellipsoid into two hemispheres.}
    \label{fig:plane}
\vspace{-2em}
\end{figure}

First, to determine whether a point is located within a spacecraft's FOV, we introduce $\epsilon$, the orthogonal distance from the axis of the FOV cone (see Fig \ref{fig:cone_fov}) to the point. Taking $O_{\mathrm{ISO}}$ to be origin of the uncertainty ellipsoid, $\epsilon$ is calculated using the following steps: 
\begin{enumerate}[label={\textcolor{uiucblue}{\arabic*}.}]
    \item the cone distance is defined as follows to measure the POI's distance from the cone's axis:
    \begin{align}
        \text{cone distance}= (p_{\mathrm{SC}}- O_{\mathrm{ISO}}) \cdot \frac{p_{\mathrm{SC}}-O_{\mathrm{ISO}}}{||p_{\mathrm{SC}}-O_{\mathrm{ISO}}||}.
    \end{align}
    \item The cone's radius at that point along the axis is then calculated as follows:
    \begin{align}
        \text{cone radius}= \frac{\text{cone distance}}{p_{\mathrm{SC}}- O_{\mathrm{ISO}}} p_{\mathrm{SC}}- O_{\mathrm{ISO}}\tan\qty(\frac{\phi}{2}).
    \end{align}
    where $\phi$ is the aperture of the camera located on the spacecraft.
    \item The orthogonal distance $\epsilon$ is then defined as: 
    \begin{align}
        \epsilon= ||(POI-p_{\mathrm{SC}}) - \text{cone distance}\times\frac{p_{\mathrm{SC}}-O_{\mathrm{ISO}}}{||p_{\mathrm{SC}}-O_{\mathrm{ISO}}||}||.
    \end{align}
\end{enumerate}
If $\epsilon$ is less than the cone radius, then the POI is said to be located within the cone. This is checked for all POIs in the hemisphere of the uncertainty ellipsoid closest to the spacecraft. 

We also introduce $\kappa$, which is a metric describing the total overlap between the FOVs of two spacecraft. The formulation of $\kappa$ is given as follows:
\begin{enumerate}[label={\textcolor{uiucblue}{\arabic*}.}]
    \item First, a spacecraft's field of view needs to be defined. $\theta$ is the initial orientation of the spacecraft's field of view (see Fig \ref{fig:cone_fov}). The FOV for a single spacecraft can be written as
    \begin{align}
        (s_{i}, e_{i})= (\theta_{i}- \nu, \theta_{i}+ \nu)
    \end{align}
    where $s_{i}$ is the start of the overlap region and $e_{i}$ is the end of the overlap region. $i=1, 2, \cdots n$, where $n$ is the number of spacecraft in the swarm. 
    \item The FOV of a second spacecraft in the swarm is defined as 
    \begin{align}
       (s_{j}, e_{j})= (\theta_{j}- \nu, \theta_{j}+ \nu)
    \end{align}
    for $j=1, 2, \cdots n$ and $j \neq i$.
    \item With these defined, we are then able to define the boundaries in which the FOVs of two spacecraft overlap: 
    \begin{align}
        {\kappa_{start}}= \max(s_i, s_j)
        {\kappa_{end}}= \min(e_i, e_j).
    \end{align}
    \item Then, the overlap is defined as
    \begin{align}
        \kappa= max(0, \kappa_{end}- \kappa_{start}).
    \end{align}
    The overlap between a spacecraft and every other spacecraft in the system is added together. This is called $\kappa_{total}$. 
\end{enumerate}
\begin{remark}
If the orientation $\theta$ between two spacecraft is the same, the FOV of the spacecraft would be the same. Then, $\kappa=0$, which is false because they would have a $100\%$ overlap. In order to overcome this issue, a small positive number is added to $\theta$ of the second spacecraft. The rest of the algorithm stays as outlined above. 
\end{remark}

With both $\epsilon$ and $\kappa$ now defined, the information cost can be formulated as follows:
\begin{align}
    \label{I}
    \mathrm{I}= \kappa_{total}-\epsilon.
\end{align}
This cost reflects our objectives of maximizing the number of POIs that are seen by a spacecraft's FOV (given by $\epsilon$) and simultaneously minimizing the amount of overlap between every spacecraft in the swarm.
\subsection{Extensions to Stochastic Terminal Positions}
We can still use our cost function~\eqref{I} even when the terminal positions of the deputy spacecraft are stochastic and generated by some probability distributions. As implied earlier in Theorem~\ref{thm_hierarchical_contraction} of Sec.~\ref{sec_problem}, this could happen, \eg{}, when the deputy spacecraft actively use the stochastic information of the ISO state and also when the dynamics~\eqref{original_dyn} and~\eqref{original_dyn_i} are subject to stochastic process noise.   a

For example, if they are assumed to be normally distributed, we could instead use the expected information cost $\mathbb{E}[\mathrm{I}]$ instead of the stochastic information cost $\mathrm{I}$ of~\eqref{I}, through integration with the probability density function of the normal distribution. When this cost is minimized, the amount of visual information the spacecraft gains from the uncertainty ellipsoid is maximized. This maximizes visual information because, ideally, each spacecraft wants to maximize the amount of POIs it covers and minimize the amount of overlap between its FOV and the FOV of another spacecraft.

\section{Numerical Simulations}
The implementation and optimization of the cost function detailed in Sec.~\ref{Info:cost}, as well as any visualizations presented in this paper, are implemented in Python. Each spacecraft in the swarm is given an initial $x$, $y$, and $z$ position, as well as an initial $\theta$. The optimization decision variables $p_{\mathrm{SC}}$ and $\theta$ are optimized using Nelder-Mead for its robustness in optimizing non-differentiable cost functions~\cite{Nelder-Mead}. Bounds are placed on $\theta$ so it always remained within $[0, 2\pi]$ radians. 
\subsection{Simulation Setup and Results}
Two experiments are performed to observe the behavior of the optimal terminal positions of the spacecraft 1. to observe the probability of one spacecraft successfully viewing the ISO and 2. to determine the optimal number of spacecraft for the system. This latter experiment is also for analyzing the trade-off between the number of spacecraft and the information gained. The success criteria are maximizing the percentage of POIs viewed by the spacecraft system and minimizing the information cost. In both experiments, we use the uncertainty ellipsoid given by Theorem~\ref{thm_hierarchical_contraction} with equal $x$, $y$, and $z$ radii with 5,000 POIs randomly distributed within the sphere. We call this an uncertainty sphere in the following.
\subsubsection{Probability of single spacecraft viewing the ISO}
The probability that a single spacecraft is able to view an ISO with varying uncertainty sphere sizes is calculated, where the positions of candidate ISOs are taken from~\cite{my_iso_mission}. Twenty-four trials are performed for three differently sized uncertainty spheres for each of two ISO terminal positions. The initial positions for the spacecraft states are randomly generated integers at 100 to 600 units away from the origin of the uncertainty sphere. The probability of the spacecraft successfully viewing the ISO is taken using the following:
\begin{align}
    p= \frac{\text{number of times ISO is viewed}}{\text{number of trials}}
\end{align}

The ISO terminal positions used are $\allowbreak:(4.1784, -9.8402, \allowbreak-4.7133)\times 10^7$ km and $(1.8804, 2.2023, 1.0029)\times 10^8$ km. The probabilities that a spacecraft can successfully view the ISO after optimization are shown in Tables~\ref{tab:tab1}~and~\ref{tab:tab2}.
\begin{table}[htbp]
\scriptsize
\centering
\renewcommand{\arraystretch}{1.3}
\begin{tabular}{ c c c c }
\hline
\hline
Uncertainty Sphere Radius & 50 & 500 & 1,000 \\
\hline
$p$ (\%) & 41.7 & 33.3 & 8.3 \\
\hline
\hline
\end{tabular}
\caption{The probability of a single spacecraft viewing the origin of an uncertainty sphere with ISO terminal position 1, where $p$ denotes the probability.}
\label{tab:tab1}
\vspace{-3em}
\end{table}
\begin{table}[htbp]
\scriptsize
\centering
\renewcommand{\arraystretch}{1.3}
\begin{tabular}{ c c c c }
\hline
\hline
Uncertainty Sphere Radius & 50 & 500 & 1,000 \\
\hline
$p$ (\%) & 25.0 & 16.7 & 25.0 \\
\hline
\hline
\end{tabular}
\caption{The probability of a single spacecraft viewing the origin of an uncertainty sphere with ISO terminal position 2, where $p$ denotes the probability.}
\label{tab:tab2}
\vspace{-3em}
\end{table}

As implied in Theorem~\ref{thm_hierarchical_contraction}, in general, the probability that a single spacecraft can view an ISO in an uncertainty sphere of varying size decreases as the radius of the uncertainty sphere increases, making the randomly generated POIs more spread out. This trend is demonstrated in the optimization experiment performed on the first ISO terminal position given in Table~\ref{tab:tab1}. However, in the simulation with the second ISO terminal position in Table~\ref{tab:tab2}, the probability that the spacecraft can view the ISO remains the same in the case where the radius is 50 and where the radius is 1,000. These observations would indicate that, due to the nonlinearity of the problem, the optimal position of the spacecraft could be dependent on the initial conditions, which causes some of the spacecraft to be positioned in such a way that they are not able to see the ISO. Additionally, limiting the initial positions to a fixed distance away from the uncertainty sphere means that as the sphere gets larger, the initial position of the spacecraft is more likely to start inside the sphere. This means that the spacecraft must find a way to optimize its position while being inside the sphere, which is not accounted for in our information cost. 
\subsubsection{Determining Optimal Number of Spacecraft}
This simulation then focuses on determining an ideal number of spacecraft for viewing the POIs in the constructed uncertainty ellipsoid. The initial positions of the spacecraft are randomly determined a large distance away from the origin of the coordinate system, which is also the center of the ellipsoid. We randomly place the spacecraft around the ellipsoid in order to reduce any bias caused by a favorable position. The information cost optimization is then run with the constraints that the angle $\theta \in [0, 2\pi)$. The information cost optimization is run from 1 to 7 spacecraft on the same uncertainty ellipsoid with the same POIs. The trial is repeated three times. 

After computing the terminal positions for each number of spacecraft, the percentage of POIs viewed by the system is determined also as a metric for performance comparison in the manner described earlier. For each simulation, we examine the information cost value of each scenario. The results are given in Table~\ref{tab:trialdata}, \ref{tab:avginfo}, \ref{tab:avgpercent}.
\begin{figure}[htbp]
    \centering
    \includegraphics[width=1\linewidth]{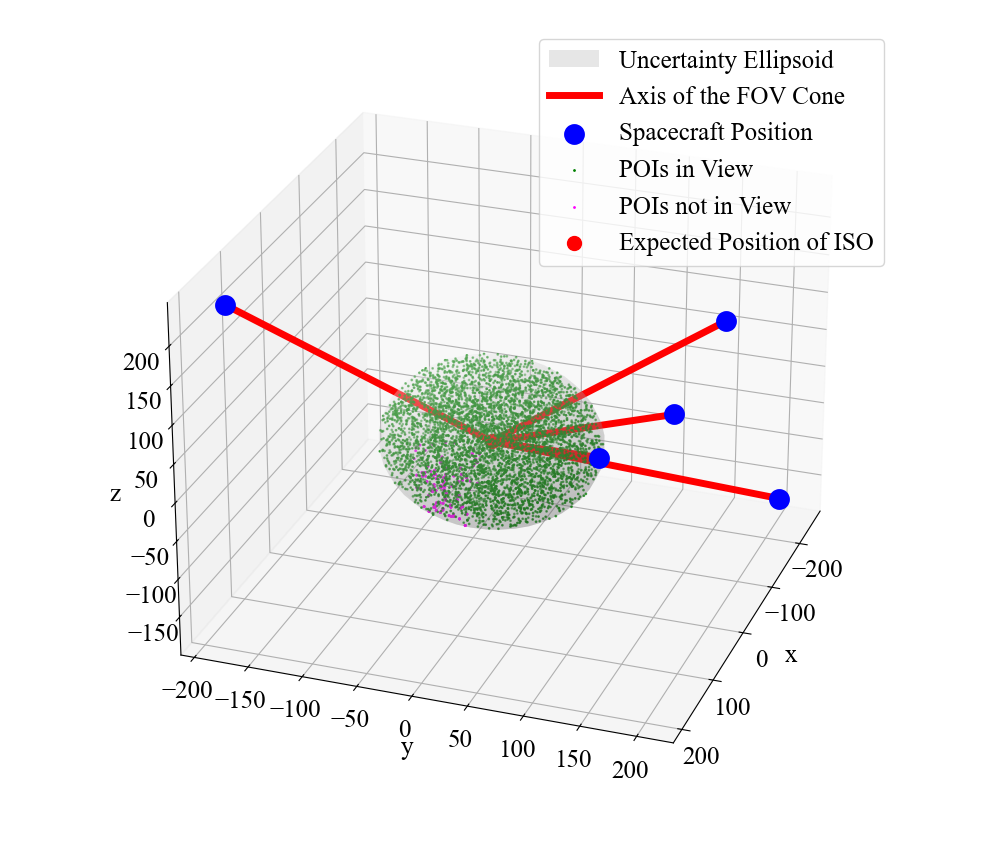}
    \caption{A diagram of the terminal positions of a five-spacecraft system and the POIs in view and not in view.}
    \label{fig:5-sc-diagram}
\end{figure}
\begin{table}[htbp]
\scriptsize
\centering
\renewcommand{\arraystretch}{1.3}
\begin{tabular}{ c c c c c c c c }
\hline
\hline
Trial \# & 1 & 2 & 3 & 4 & 5 & 6 & 7 \\
\hline
Trial 1 & 0.45 & 0.65 & 0.82 & 0.93 & 0.95 & 1.00 & 1.00 \\
Trial 2 & 0.50 & 0.77 & 0.98 & 1.00 & 0.98 & 0.99 & 1.00 \\
Trial 3 & 0.50 & 0.64 & 0.72 & 0.51 & 0.85 & 0.90 & 0.96 \\
\hline
\hline
\end{tabular}
\caption{Percentage of 5,000 POIs observed with respect to the number of spacecraft across the three trials.}
\label{tab:trialdata}
\end{table}

Table~\ref{tab:trialdata} implies that, locally, the optimal number of spacecraft for an uncertainty sphere of radius 100 units is five spacecraft. The system views 92.740\% of the POIs across three trials while producing the lowest information cost in the simulations.
All trials follow a similar quasi-logarithmic increase in coverage,  converging at 100\% observation in Table~\ref{tab:trialdata}. Trial 3's four-spacecraft system has a significantly reduced coverage of the POIs, which could be due to unfavorable random initial conditions. All other trials have a similar trend. Also, as indicated in Table~\ref{tab:avgpercent}, as the number of spacecraft increases, the percentage of POIs observed also increases; however, this does not correlate to lower information costs. As shown in Table~\ref{tab:avginfo}, the information cost increases significantly after three trials due to the increased overlap between the spacecraft FOVs. This directly leads to a sharp increase in the information cost. 

\begin{table}[htbp]
\scriptsize
\centering
\renewcommand{\arraystretch}{1.3}
\begin{tabular}{ c c c c c c c c }
\hline
\hline
\# S/C & 1 & 2 & 3 & 4 & 5 & 6 & 7 \\
\hline
-$I$  & $95.1$ & $95.3$ & $96.6$ & $95.6$ & $96.1$ & $82.9$ & $69.0$ \\
\hline
\hline
\end{tabular}
`   \caption{The average information cost with respect to the number of spacecraft across the three trials where \# S/C is the number of spacecraft in the system and $-I$ is the negative information cost.}
\label{tab:avginfo}
\end{table}

\begin{table}[htbp]
\scriptsize
\centering
\renewcommand{\arraystretch}{1.3}
\begin{tabular}{ c c c c c c c c }
\hline
\hline
\# S/C & 1 & 2 & 3 & 4 & 5 & 6 & 7 \\
\hline
$p$ (\%) & $50.2$ & $68.6$ & $83.9$ & $81.6$ & $92.7$ & $96.7$ & $98.6$ \\
\hline
\hline
\end{tabular}
\caption{The average percentage of 5000 POIs viewed by the system, where \# S/C denotes the number of spacecraft and $p$ denotes the percentage of POIs observed.}
\label{tab:avgpercent}
\end{table}

This experiment displays the benefits of a multi-spacecraft system due to its increased observation of the number of POIs and gives insight into trade-offs of more spacecraft in the system. The objective of the simulation is not to prove that a five-spacecraft system is the ideal configuration for a flyby mission but to empirically show the potential of a multi-spacecraft system in the ISO encounter, with a saturation point in which an excess of spacecraft could lead to a waste of resources. There could be a necessity to have a greater number of spacecraft for an uncertainty ellipsoid of greater radius, and the results presented here imply that multi-spacecraft systems could perform better than a singular spacecraft in gathering valuable scientific data during the encounter. 
\section{Conclusion}
This paper presents a multi-spacecraft framework for optimizing information gain during ISO encounters, addressing the challenges posed by the large state uncertainty. By leveraging hierarchical stochastic contraction, we first derive an expected spacecraft delivery error bound for constructing an uncertainty ellipsoid, where the ISO would be located with a finite probability. We then design a novel cost function to maximize the visual information gained from an ISO by a swarm of spacecraft, explicitly considering the spacecraft specifications. Numerical simulations are performed for different uncertainty ellipsoids and different numbers of spacecraft in the swarm. It can be observed that each spacecraft in the swarm locally maximizes the amount of the uncertainty ellipsoid it views while minimizing the amount of overlap between its FOV and the FOV of other spacecraft in the swarm. ISOs follow a hyperbolic trajectory and will never re-enter the solar system following their exit. Our proposed approach could help us gain as much visual information from them as possible through such rare events, thereby providing additional tools for furthering the understanding of deep-space celestial objects and interstellar space. 



\subsection{Potential Future Work}
This project can be extended by creating optimal orbits for each of the spacecraft in the swarm to orbit the ISO's uncertainty ellipsoid. Following the ISO for an extended period would enable the swarm to gain information from many perspectives. Also, further work can be done to implement this cost function in a swarm of spacecraft that is observing multiple ISOs. There is a chance that multiple ISOs enter the solar system at similar times or that an ISO may splinter into two or more different parts. Because of this, the swarm must be able to optimally capture visual data from these multiple ISOs. Finally, the cost function presented in this paper could be extended to flyby missions involving asteroids, comets, or other celestial bodies. Future studies would include comparing the results of optimization using this algorithm of ISOs and other celestial bodies. 
\acknowledgments
Part of the research was carried out at the Jet Propulsion Laboratory, California Institute of Technology, under a contract with the National Aeronautics and Space Administration.
\bibliographystyle{IEEEtran}
\bibliography{bib.bib}
\thebiography
\vspace{-1em}
\begin{biographywithpic}
{Arna Bhardwaj}{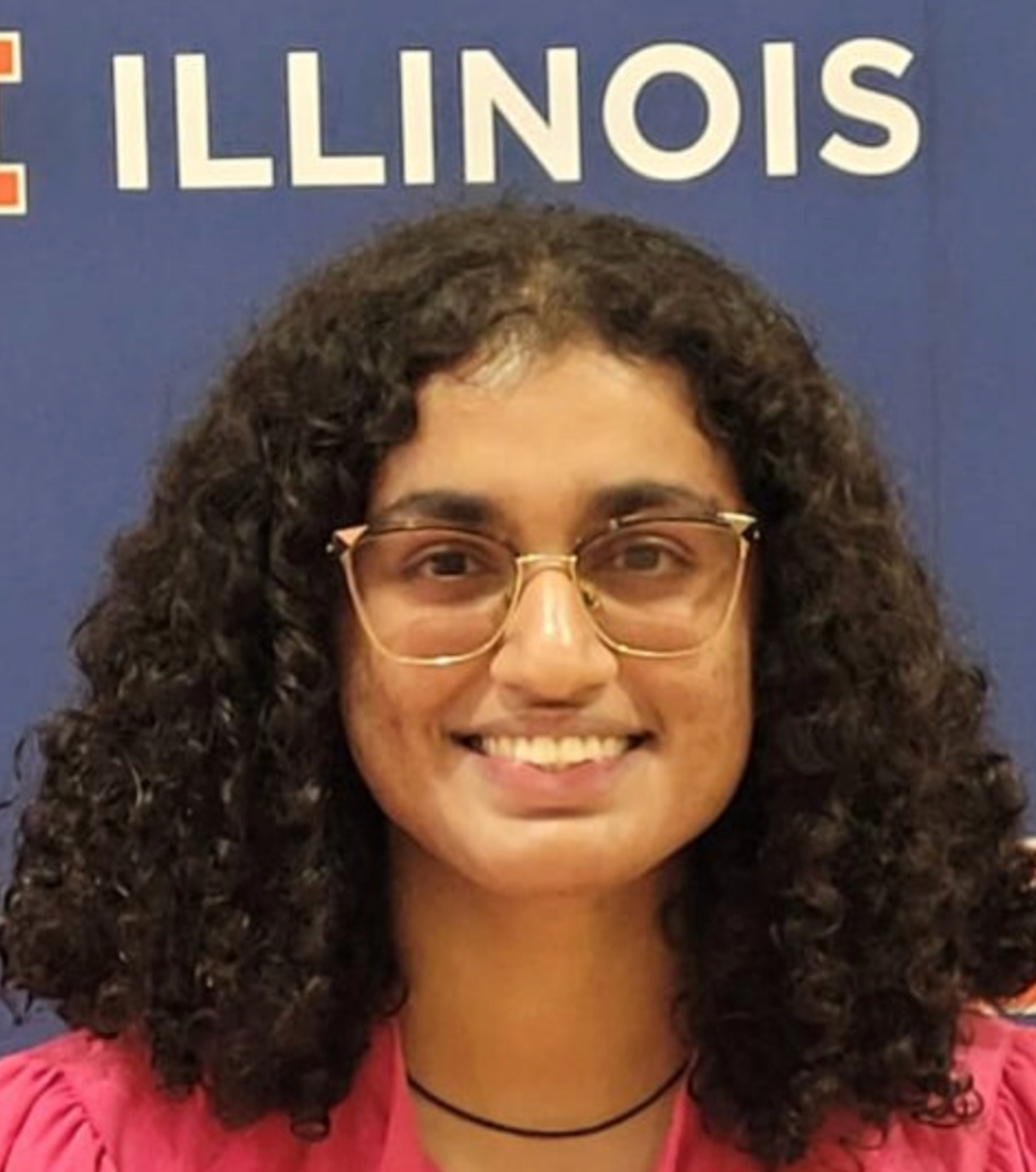}
will receive her B.S. in Aerospace Engineering with a Minor in Statistics from the University of Illinois at Urbana-Champaign in May 2025. She is interested in both perception and estimation for stochastic systems. Throughout her four years at Illinois, Arna has been a part of four different research labs and spent two summers as a NASA Space Grant Consortium researcher under their Undergraduate Research Opportunities Program. She hopes to pursue a PhD in Aerospace Engineering starting in Fall 2025 and pursue her goal of developing robots to aid in human exploration of extraterrestrial domains. 
\end{biographywithpic} 

\begin{biographywithpic}
{Shishir Bhatta}{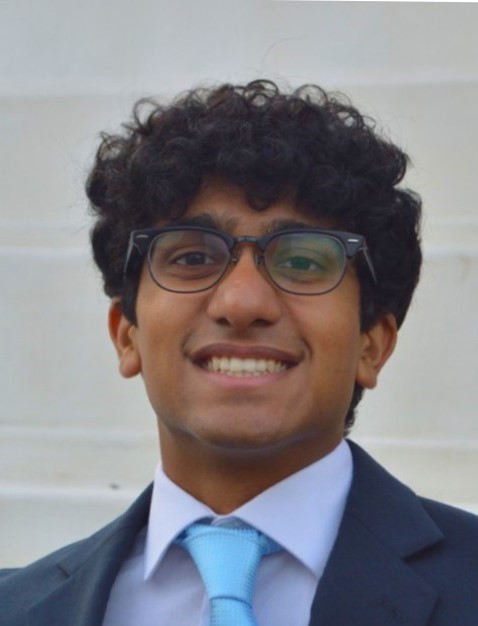}
is studying for his B.S. in Aerospace Engineering with a minor in Computer Science at the University of Illinois Urbana-Champaign (UIUC), expecting to graduate in 2026. He is interested in the cross-section of guidance, navigation and controls (GNC), and flight software, exploring these topics through his academics and extracurricular activities. He has previously worked with the Advanced Controls Research Lab at UIUC, developing a gradient-descent-based framework to tune controller parameters. He previously interned at AstroForge as a Flight Software Intern, exploring his interests in flight software. He is currently working with the ACXIS Lab as an undergraduate researcher. 

\end{biographywithpic}

\begin{biographywithpic}
{Hiroyasu Tsukamoto}{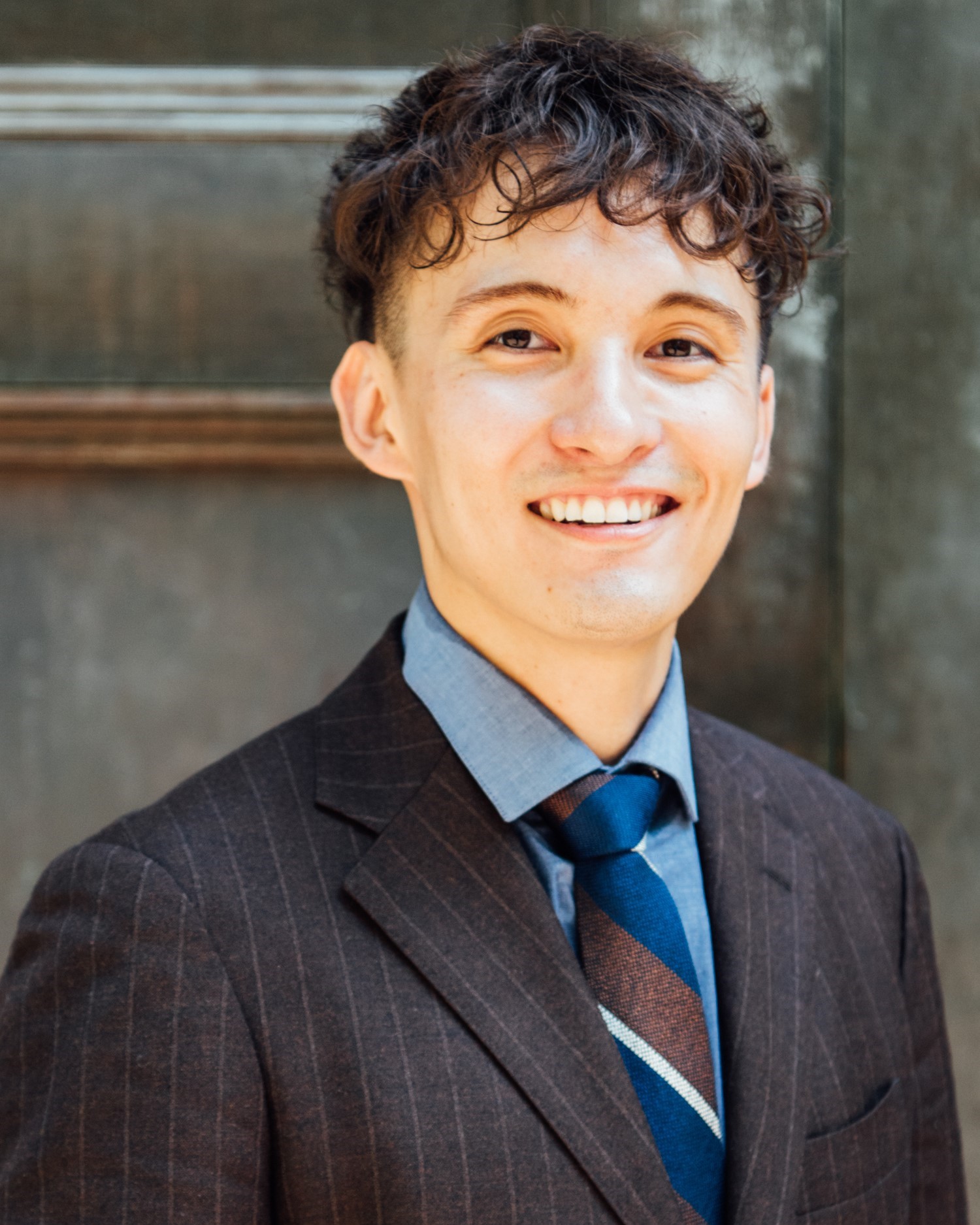} is an Assistant Professor of Aerospace Engineering at the University of Illinois at Urbana-Champaign and the director of the ACXIS Laboratory (Autonomous Control, Exploration, Intelligence, and Systems). Prior to joining Illinois, he was a Postdoctoral Research Affiliate in Robotics at the NASA Jet Propulsion Laboratory, where he contributed to the Science-Infused Spacecraft Autonomy for Interstellar Object Exploration and Multi-Spacecraft Autonomy Technology Demonstration projects. He received his Ph.D. and M.S. in Space Engineering (Autonomous Robotics and Control) from Caltech in 2018 and 2023, respectively, and his B.S. degree in Aeronautics and Astronautics from Kyoto University, Japan, in 2017. He is the recipient of several awards, including the William F. Ballhaus Prize for the Best Doctoral Dissertation in Space Engineering at Caltech and the Innovators Under 35 Japan Award from MIT Technology Review. More info: \myhref{https://hirotsukamoto.com}{https://hirotsukamoto.com}.
\end{biographywithpic}

\end{document}